\documentclass[aps,pra,twocolumn,showpacs]{revtex4-1}
\usepackage{amsfonts}
\usepackage{amsmath}
\usepackage{amssymb}
\usepackage{amsthm}
\usepackage{graphicx}
\usepackage{xspace}

\newcommand{\QSL}{QSL\xspace}
\newcommand{\QSLs}{\QSL{\lowercase{s}}\xspace}
\newcommand{\aveE}[1]{\langle E^{#1} \rangle}
\newcommand{\absE}[1]{\langle |E|^{#1} \rangle}

\newtheorem{lemma}{Lemma}

\begin{document}
\title{Quantum Speed Limit With Forbidden Speed Intervals}
\author{H. F. Chau}
\email{hfchau@hku.hk}
\affiliation{Department of Physics and Center of Theoretical and Computational
 Physics, Pokfulam Road, Hong Kong}
\date{\today}

\begin{abstract}
 Quantum mechanics imposes fundamental constraints known as quantum speed
 limits (\QSLs) on the information processing speed of all quantum systems.
 Every \QSL known to date comes from the restriction imposed on the evolution
 time between two quantum states through the value of a single system
 observable such as the mean energy relative to its ground state.
 So far these restrictions only place upper bounds on the information
 processing speed of a quantum system.
 Here I report \QSLs each with permissible information processing speeds
 separated by forbidden speed intervals.
 They are found by a systematic and efficient procedure that takes the values
 of several compatible system observables into account simultaneously.
 This procedure generalizes almost all existing \QSL proofs; and the new \QSLs
 show a novel first-order phase transition in the minimum evolution time.
\end{abstract}


\maketitle

\section{Introduction} \label{Sec:Intro}
 Time is a valuable and often an irreplaceable resource.  Minimizing runtime is
 one of the most important driving forces behind hardware, software and
 computational complexity researches.  Quantum mechanics, as a fundamental law
 of nature, gives ultimate constraints known as \QSLs on the runtime and hence
 information processing speed of a computer, classical and quantum
 alike~\cite{Lloyd00,PHYSCOMP96,ML98}.
 \QSLs can be defined by considering the distance between a normalized initial
 state $|\psi\rangle$ and the normalized state $|\varphi\rangle$ that evolves
 from it after a time $\tau$, as given by their mutual fidelity $F = \left|
 \langle\psi|\varphi\rangle \right|^2$.
 Since the speed of evolution is determined by the energy of the system, there
 exists \QSLs in the form $\tau \ge g_{\hat{\mathcal O}}(F,v)$, where $v$ is
 the expectation value of an observable $\hat{\mathcal O}$ associated with the
 energy of the
 system~\cite{PHYSCOMP96,ML98,GLM03,MT45,Fleming73,Bhattacharyya83,Uhlmann92,%
Pfeifer93,ZZ06,JK10,Chau10,LC13}.
 The significance of such \QSLs is that they can relate the evolution time
 needed to achieve a given distance between initial and final states to only a
 single observable property of the system, allowing an efficient evaluation of
 the physical resources necessary to achieve maximum quantum information
 processing speed.
 For example, using the energy standard deviation $\Delta E$ as the observable,
 the corresponding \QSL, namely, $\tau \ge \hbar \cos^{-1} (\sqrt{F}) / \Delta
 E$, is the famous time-energy uncertainty
 relation~\cite{MT45,Fleming73,Bhattacharyya83,Uhlmann92,Pfeifer93}.

 The more information about the quantum system is given, the more stringent the
 \QSL one should be able to obtain.
 At one of the extreme ends that nothing is known about the system, the only
 thing one can say is the trivial bound $\tau \ge 0$.
 At the other extreme end that the Hamiltonian and initial state are completely
 known, the values of all the $\tau$'s are fixed and can be computed at least
 in principle.
 Thus, it is instructive to investigate what kind of \QSLs one can deduce when
 partial information in the form of more than one observable of the system are
 given.  In addition to quantum mechanics and quantum information, this problem
 is also of interest in statistical physics.
 For example, the constraints and information on $\tau$ and their related phase
 diagrams as a function of the kind and amount of information given are
 important questions that have never been studied.
 In fact, very limited progress has been made along these lines.  All relevant
 works to date simply consider the constraints that logically follow from the
 \QSLs for each of the observables~\cite{GLM03,LT09} instead of analyzing the
 restrictions due to these compatible observables holistically.

 Here I report a powerful method to study \QSLs by extending earlier proof
 techniques~\cite{GLM03,Chau10,LC13}.  Besides finding \QSLs when several
 compatible observables of the system are given, this method also provides a
 unified way to show all known \QSLs.  Through these new \QSLs, I find that the
 minimum possible evolution time can exhibit a new first-order phase transition
 with fidelity $F$ being the order parameter.
 Finally, I will define and briefly discuss the reverse problem of \QSL
 construction.

\section{Construction Of The \QSL}
\label{Sec:Main}
\subsection{An auxiliary inequality}
\label{Subsec:auxiliary}
 I write the time-independent Hamiltonian $H$ in the diagonal form $\sum_j E_j
 |E_j\rangle\langle E_j|$.  Surely, any normalized quantum state $|\psi\rangle$
 can be expressed in the form $\sum_j \alpha_j |E_j\rangle$ with $\sum_j
 |\alpha_j|^2 = 1$.  Suppose $|\varphi\rangle = e^{-i H \tau / \hbar}
 |\psi\rangle$ is the state obtained by evolving $|\psi\rangle$ by $H$ for a
 time $\tau \ge 0$.  Then, the fidelity $F$ between $|\psi\rangle$
 and $|\varphi\rangle$ obeys
\begin{align}
 \sqrt{F} & = \left| \sum_j |\alpha_j|^2 e^{-i E_j \tau / \hbar} \right| =
  \left| e^{i\theta} \right| \left| \sum_j |\alpha_j|^2 e^{-i E_j \tau / \hbar}
  \right| \nonumber \\
 & \ge \sum_j |\alpha_j|^2 \cos \left( \frac{E_j \tau}{\hbar} - \theta \right)
 \label{E:general_inequality}
\end{align}
 for any real-valued $\theta$.  Here I have used the fact that the magnitude of
 a complex number is greater than or equal to its real part to arrive at the
 above inequality.  Actually, inequality~\eqref{E:general_inequality} is an
 extension of those used in Refs.~\cite{Chau10,LC13}.

 Let $p(x)$ be a function satisfying $p(x) \le \cos (x-\theta)$ whenever $x\ge
 0$.  Then,
\begin{equation}
 \sqrt{F} \ge \sum_j |\alpha_j|^2 p(E_j \tau / \hbar) \label{E:pre_inequality}
\end{equation}
 provided that $E_j \ge 0$ for all $j$.
 The following subsection shows that this $p(x)$ can be chosen to be a
 polynomial-like function in the form $\sum_{k=0}^n c_k x^{s k}$ with $s > 0$
 efficiently.  (This is a more general choice for $p(x)$ than in all previous
 studies~\cite{PHYSCOMP96,ML98,GLM03,ZZ06,Chau10,LC13}.)

\subsection{Proof of the existence of efficiently computable
 $\boldsymbol{p(x)}$}
\label{Subsec:efficiency}
 It suffices to show the existence of a polynomial $q(x) \le \cos (x^{1/s} -
 \theta) \equiv f(x)$ whenever $x\ge 0$.  In fact, $q(x)$ exists even if I
 further demand that it meets $f(x)$ at finitely many distinct non-negative
 points, say, $x_\ell$'s.
 I do this by considering the Hermite interpolating polynomial $\tilde{q}(x)$
 that satisfies $\tilde{q}^{(k)}(x_\ell) = f^{(k)}(x)$ for $k=0,1,\ldots,
 2j_\ell-1$ for some $j_\ell \in {\mathbb Z}^+$.  This polynomial can be
 constructed efficiently~\cite{SB02}.  Surely, $f(x) - \tilde{q}(x) = a_\ell
 (x-x_\ell)^{2j_\ell} + \text{O}((x-x_\ell)^{2j_\ell+1})$ locally around each
 $x_\ell$.  For randomly chosen $x_\ell$'s and $j_\ell$'s, all the $a_\ell$'s
 are non-zero almost surely.  And in the singular case in which some of the
 $a_\ell$ equals $0$, I simply randomly choose an extra distinct point $x_a$
 and demand further that $\tilde{q}(x)$ obeys $\tilde{q}^{(k)}(x_a) =
 f^{(k)}(x_a)$ for $k=0,1,\ldots, 2j_a-1$ for some randomly chosen positive
 integer $j_a$.  Then, the modified $\tilde{q}(x)$ locally agrees with $f(x)$
 up to an even power of $x-x_\ell$ at each of the interpolating point $x_\ell$
 almost surely.

 Recall that there are efficient and stable algorithms to find all the real
 roots of a polynomial~\cite{McNamee07}.  Apply one such algorithm to find the
 largest real root $x_u$ of the equation $\tilde{q}(x) = -1$, I can bound the
 non-negative roots of $\tilde{q}(x) = f(x)$ to the interval $[0,x_u]$.  Since
 $f(x) = \cos (x^{1/s} - \theta)$ is a smooth function of bounded variation in
 $[0,x_u]$, I may use interval arithmetic to efficiently find all the
 sub-intervals of $[0,x_u]$, if any, in which $f(x) - \tilde{q}(x)$ is
 negative~\cite{AH83}.  Actually, there are at most a finite number of these
 sub-intervals; and they are present if and only if
\begin{itemize}
 \item $f(x) - \tilde{q}(x) = a_\ell (x - x_\ell)^{2j_\ell} + \text{O}((x-
 x_\ell)^{2j_\ell+1})$ in the neighborhood of $x_\ell$ with $a_\ell < 0$; or
 \item for two consecutive distinct roots $x_1$ and $x_2$ of $f(x) -
  \tilde{q}(x)$, there are $x_1 < y_1 < y_2 < x_2$ such that $f(x) -
  \tilde{q}(x) < 0$ for all $x \in (y_1,y_2)$.
\end{itemize}
 In the first case, I may bring $f(x) - \tilde{q}(x)$ up above zero by adding
 a term $b_\ell (x-x_\ell)^{2j_\ell}$ with $b_\ell > a_\ell$.  Whereas in the
 second case, this can be done by adding a term in the form $b' \prod_\ell
 (x-x_\ell)^{2(j_\ell + \kappa_\ell)}$ with $\kappa_\ell = 1$ if $a_\ell > 0$
 and $\kappa_\ell = 0$ if $a_\ell < 0$.  Note that this additional term does
 not affect the local behavior of $f(x) - \tilde{q}(x)$ around those $x_\ell$'s
 with $a_\ell > 0$.  Since there are only a finite number of such
 sub-intervals, I can efficiently find $b > 0$ such that $f(x) - q(x) \ge 0$
 for all $x\ge 0$ where $q(x) = \tilde{q}(x) - b \prod_\ell
 (x-x_\ell)^{2(j_\ell+\kappa_\ell)}$.

 Last but not least, I remark that since what one really need is $\tilde{q}(x)
 \le f(x)$ for $x \ge 0$.  So, whenever $x_\ell = 0$, namely, the boundary
 point, there is no need to demand $f(x) - \tilde{q}(x) = a_\ell x^{2j_\ell} +
 \text{O}(x^{2j_\ell+1})$ for $x$ sufficiently close to $0$.  Suppose $f(x) -
 \tilde{q}(x) = a x^j + \text{O}(x^{j+1})$ in the neighborhood of $x = 0$ for
 some positive integer $j$ and $a \ne 0$.  Then, the transformation $q(x) =
 \tilde{q}(x) - b x^{j+\kappa} \prod_\ell' (x-x_\ell)^{2(j_\ell +
 \kappa_\ell)}$ for a sufficiently large $b > 0$ will do.  Here $\kappa = 1$ if
 $a > 0$ and $\kappa = 0$ otherwise.  Besides, the primed product is over all
 $x_\ell \ne 0$.

\subsection{Construction of the \QSL from $\boldsymbol{p(x)}$}
\label{Subsec:QSL_Construction}
 Substituting this polynomial-like $p(x)$ into
 inequality~\eqref{E:pre_inequality}, I conclude that
\begin{subequations}
\label{E:inequality_set}
\begin{equation}
 \sqrt{F} \ge \sum_{k=0}^n c_k \aveE{s k} \left( \frac{\tau}{\hbar} \right)^{s
 k} ,
 \label{E:inequality_ave_E}
\end{equation}
 whenever $E_j \ge 0$ for all $j$, where $\aveE{r} \equiv \sum_j |\alpha_j|^2
 E_j^r$ denotes the expectation value of the $r$th moment of the energy of the
 system.  Furthermore, in the case of $\theta = 0$, I may rewrite
 inequality~\eqref{E:general_inequality} as $\sqrt{F} \ge \sum_j |\alpha_j|^2
 \cos ( |E_j| \tau / \hbar )$.  So the above arguments lead to
\begin{equation}
 \sqrt{F} \ge \sum_{k=0}^n c_k \absE{s k} \left( \frac{\tau}{\hbar} \right)^{s
 k}
 \label{E:inequality_abs_E}
\end{equation}
\end{subequations}
 irrespective of the signs of $E_j$'s, where $\absE{r} \equiv \sum_j
 |\alpha_j|^2 |E_j|^r$.  I remark that inequality~\eqref{E:inequality_set}
 becomes an equality if and only if $e^{-i\theta} \langle\psi|\varphi\rangle$
 is real and non-negative together with $\cos (E_j \tau / \hbar - \theta) =
 p(E_j \tau / \hbar)$ for all $j$ with $\alpha_j \ne 0$.

 Consequently, suppose the values of the compatible (time-independent)
 observables of the system $\aveE{k}$ (or $\absE{k}$) are known for $k=1,2,
 \ldots,n$.  Then, given a fixed fidelity $F\in [0,1]$, the required evolution
 time must satisfy inequality~\eqref{E:inequality_set}.  Since there are
 efficient numerical algorithms to find real roots of a polynomial
 equation~\cite{Pan97}, I can find the permissible intervals for $\tau$
 readily.  As the reference energy level has no physical meaning, I may tighten
 the permissible region for $\tau$ by taking the intersection over all the
 permissible intervals obtained by replacing $\aveE{k}$ by $\langle (E+a)^k
 \rangle$ (or $\absE{k}$ by $\langle |E+a|^k \rangle$) for all $a$ and $k$ ---
 a trick first used in Ref.~\cite{Chau10}.  Finally, I may further strengthen
 the bound by taking the intersection over the permissible regions for $\tau$
 obtained by all degree $\le n$ polynomials $p(x) \le \cos (x-\theta)$ for $x
 \ge 0$.  There is no known efficient way to perform this very last task,
 however.

\subsection{Recovering all existing \QSLs}
\label{Subsec:recovering_existing_QSL}
 The above procedure, even without the final step, is already powerful enough
 to prove all the known \QSLs.
 I choose $p(x)$ to be the function $1 - a x^b$ which meets the curve $\cos
 (x-\theta)$ at two distinct points for $x\ne 0$, namely, at $x = 0$ and $x_c$
 such that $p(x)$ actually touches the curve $\cos (x-\theta)$ tangentially at
 the latter point.  Then, in the event that $\theta\ne 0$,
 inequality~\eqref{E:inequality_ave_E} gives the
 Margolus-Levitin bound~\cite{PHYSCOMP96,ML98,GLM03} and its
 generalization~\cite{ZZ06}; whereas in the event that $\theta = 0$,
 inequality~\eqref{E:inequality_abs_E} becomes the Chau bound~\cite{Chau10} and
 its generalization~\cite{LC13}.

 To recover the time-energy uncertainty relation, I use the inequality $\cos x
 \ge 1 - x^2 / 2$.  From inequality~\eqref{E:inequality_abs_E}, I get the bound
 $\sqrt{F} \ge 1 - \absE{2} (\tau / \hbar)^2 / 2 = 1 - \aveE{2} (\tau /
 \hbar)^2 / 2$.  This bound can be optimized by choosing the reference energy
 level to be the average energy of the system.  The result is $\sqrt{F} \ge 1 -
 (\tau \Delta E / \hbar)^2 /2$ provided that the system evolves under a
 time-independent Hamiltonian.   Now I consider evolving the system for an
 infinitesimal time $\Delta \tau$.  The constraint set by the above inequality
 becomes $\sqrt{F} \ge \cos (\Delta\tau \Delta E / \hbar) +
 \text{O}((\Delta\tau)^4) = \cos (\Delta\tau \Delta E / \hbar +
 \text{O}((\Delta\tau)^3))$.  Hence, the corresponding infinitesimal change in
 Bures angle must obey $\Delta\vartheta \equiv \cos^{-1}(\sqrt{F}) \le
 \Delta\tau \Delta E / \hbar + \text{O}((\Delta\tau)^3)$.  Since Bures angle is
 a metric~\cite{Uhlmann95}, by integrating over a finite time, I conclude that
 for a time-dependent Hamiltonian, the evolution time $\tau$ obeys $\int_0^\tau
 \Delta E \ d\tau \ge \hbar \vartheta = \hbar \cos^{-1} ( \sqrt{F})$, which is
 the time-energy uncertainty relation for time-dependent Hamiltonian.  If the
 Hamiltonian is time-independent, the above expression becomes the famous
 inequality $\tau \ge \hbar \cos^{-1}(\sqrt{F}) / \Delta E$.

\section{Existence of \QSLs with forbidden speed intervals}
\label{Sec:FSI}
\subsection{General discussions}
\label{Subsec:FSI_general}
 Note that a degree greater than one polynomial is in general not monotonic.
 Thus, the domain for such a polynomial to be greater than or equal to a
 certain fixed given value is in general consists of finite number of
 intervals.
 Thus, by picking the polynomial-like function $p(x)$ with $n > 1$, I have the
 surprising situation
 that the permissible evolution time $\tau$ given by
 inequality~\eqref{E:inequality_set} is in general separated by forbidden time
 intervals.  This is not completely unexpected because unlike all previous
 \QSLs, here the quantum system is constrained by more than one compatible
 observables.

 One may question the genuineness of these forbidden time intervals as some of
 the apparently permissible time intervals are illusory because they could be
 the result of poorly chosen $\theta$ and $p(x)$.
 In other words, perhaps these so-called forbidden time intervals will
 disappear once a \QSL is obtained from a carefully picked $\theta$ and $p(x)$.
 Nevertheless, the example below shows the contrary.

 Consider the initial state $|\varphi_e(0)\rangle = [\sqrt{7} |0\rangle +
 \sqrt{2} ( |\mathord{-}\hbar\rangle + |\hbar\rangle +
 |\mathord{-}11\hbar/5\rangle + |11\hbar/5\rangle ) ] / \sqrt{15}$ evolving
 under the time-independent Hamiltonian $H_e = \sum_{j=0,\pm 1,\pm 11/5}$
 $\hbar E_j|\hbar E_j\rangle\langle\hbar E_j|$.  Fig.\ref{F:example}a depicts
 the time evolution curve for the root fidelity $\sqrt{F} = \left|
 \langle\varphi_e(0)|\varphi_e(\tau)\rangle \right|$, showing that the first
 time for $\sqrt{F}$ to reach $\sqrt{F_{c1}} = 0$ and $\sqrt{F_{c2}} \approx
 0.0682$ are at $\tau = \tau_{c1} \approx 9.693$ and $\tau_{c2} \approx 4.110$,
 respectively.
 I choose $p(x) = p_e(x)$ to be the Hermite interpolating polynomial satisfying
 the following constraints: $p_e^{(j)}(0) \equiv d^j p_e (0)/dx^j =
 \cos^{(j)}(0)$ for $j = 0,1,2$, $p_e^{(j)}(\pm\tau_{c1}) = \cos^{(j)}
 (\pm\tau_{c1})$ for $j = 0,1,2,3$, $p_e^{(j)}(\pm\tau_{c2}) = \cos^{(j)}
 (\pm\tau_{c2})$ for $j = 0,1$, $p_e^{(j)}(\pm 11\tau_{c1}/5) = \cos^{(j)}
 (\pm 11\tau_{c1}/5)$ for $j = 0,1,2,3,4,5$, and $p_e^{(j)}(\pm 11\tau_{c2}/5)
 = \cos^{(j)}(\pm 11\tau_{c2}/5)$ for $j = 0,1,2,3$.  By construction, $p_e(x)
 = \cos x$ at $x \in {\mathcal S}_e = \{0,\pm\tau_{c1},\pm\tau_{c2},\pm
 11\tau_{c1}/5, \pm 11\tau_{c2}/5 \}$.  By construction, $p_e(x)$ is even and
 of degree~34.  Fig.~\ref{F:example}b shows that $p_e(x)$ is a very good
 approximation to $\cos x$ for $0\leq x \lesssim 22$.  More importantly, I show
 in the next subsection that $p_e(x) \le \cos x$ for all real $x$.

\subsection{Proof of $\boldsymbol{p_e(x) \le \cos x}$}
\label{Subsec:proof_of_p_e_inequality}
 It is straightforward to check that $p_e(x) \ge \cos x$ in the neighborhood of
 $x \in {\mathcal S}_e$.  To show that this is also true for all $x \in
 {\mathbb R}$, I need the following lemma.

\begin{lemma} \label{Lem:Rolle}
 Let $f(x) \colon {\mathbb R} \to {\mathbb R}$ be a real-valued differentiable
 function with exactly $n$ real roots counted by multiplicity.  Then, $f'(x)$
 has at least $n-1$ real roots counted by multiplicity.
\end{lemma}
\begin{proof}
 The lemma is a simple consequence of the following two facts.  First, if $x_1,
 x_2$ are two distinct consecutive roots of $f$, then Rolle's theorem implies
 that there is a root $\tilde{x} \in [x_1,x_2]$ for $f'$.  Second, suppose
 $x_1$ is a multiple root of $f$ of multiplicity $k > 1$, then clearly $x_1$ is
 a root of $f'$ of multiplicity $k-1$.
\end{proof}

 By construction, $x = 0$ is a root of multiplicity~$3$ for the even function
 $g(x) = p_e(x) - \cos x$.  Similarly, by counting the multiplicity of roots of
 $g(x)$ at $x\in {\mathcal S}_e \setminus \{ 0 \}$, I conclude that $g(x)$ has
 at least~$35$ real roots.  Suppose it had more than~$35$ such roots, then the
 even function $g(x)$ should have at least two more real roots --- one
 positive, one negative.  By Lemma~\ref{Lem:Rolle}, $g^{(32)}(x) =
 p_e^{(32)}(x) - \cos x$ would have at least~$35+2-32 = 5$ real roots.
 However, a plot of $\cos x$ and the quadratic function $p_e^{(32)}(x)$ in
 Fig.~\ref{F:example}c shows that $g^{(32)}(x)$ only has~$4$ roots counting by
 multiplicity in the range $x\in [-8,8]$; and $g^{(32)}(x)$ does not have any
 real root outside this range as $p_e^{(32)}(x) < -1$ for $|x| > 8$.  Thus,
 ${\mathcal S}_e$ is the set of all real roots of $g(x)$.  Since $g(x) \le 0$
 in the neighborhood of these roots, the continuity of $g$ implies that $g(x)
 \le 0$ for all $x\in {\mathbb R}$.

\begin{figure*}[t]
 \centering\includegraphics[width=7.5cm]{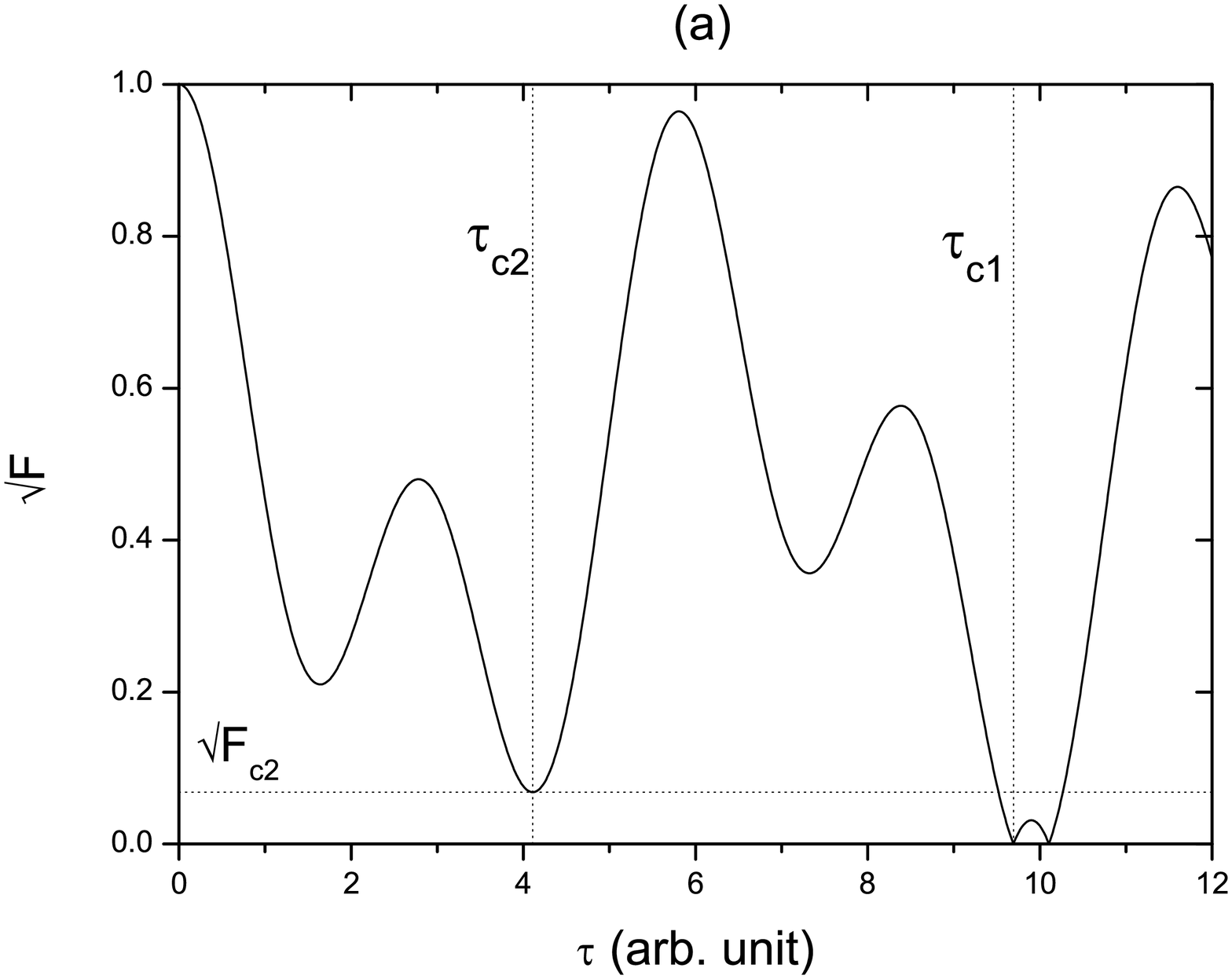}
 \centering\includegraphics[width=7.5cm]{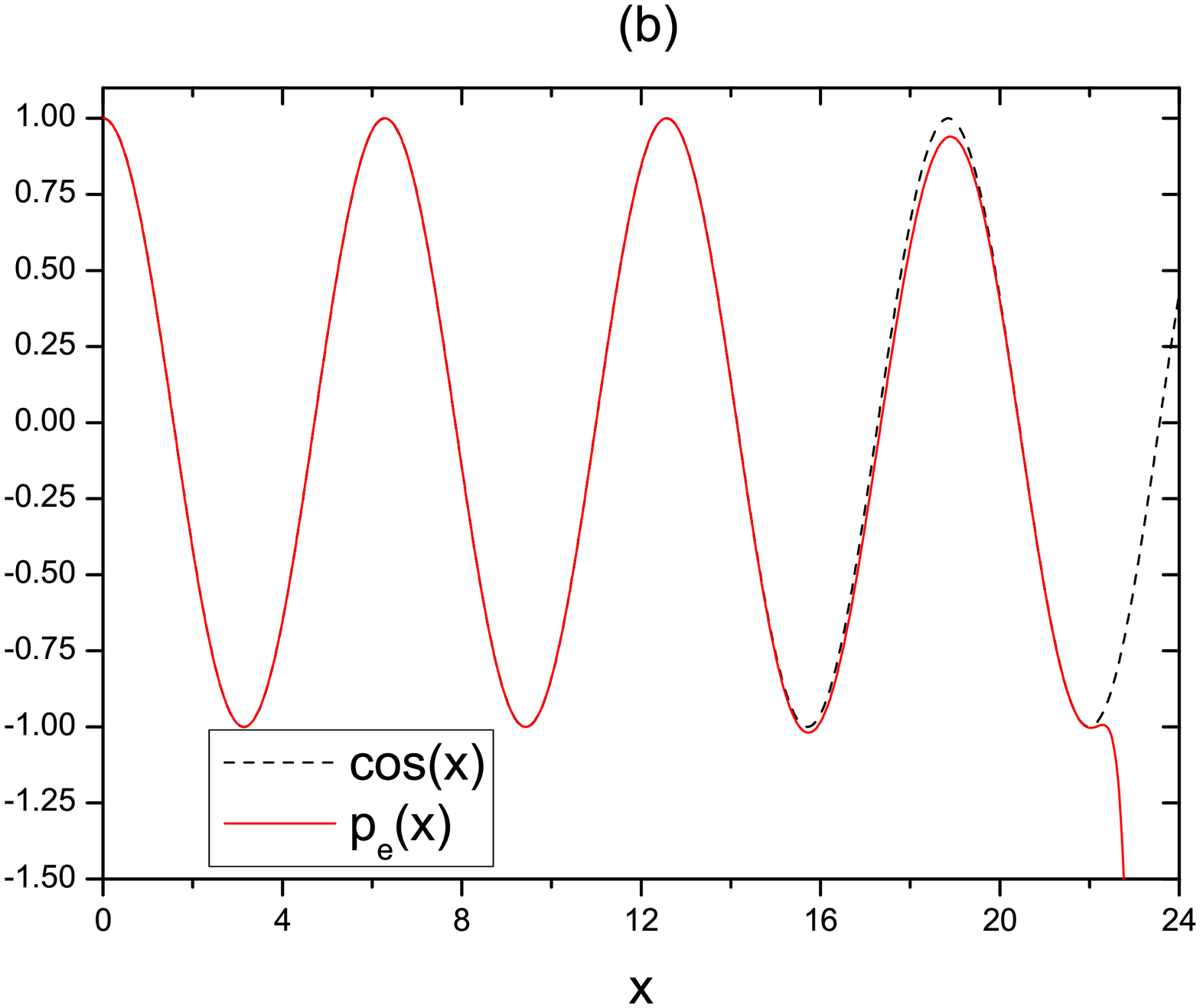}
 \centering\includegraphics[width=7.5cm]{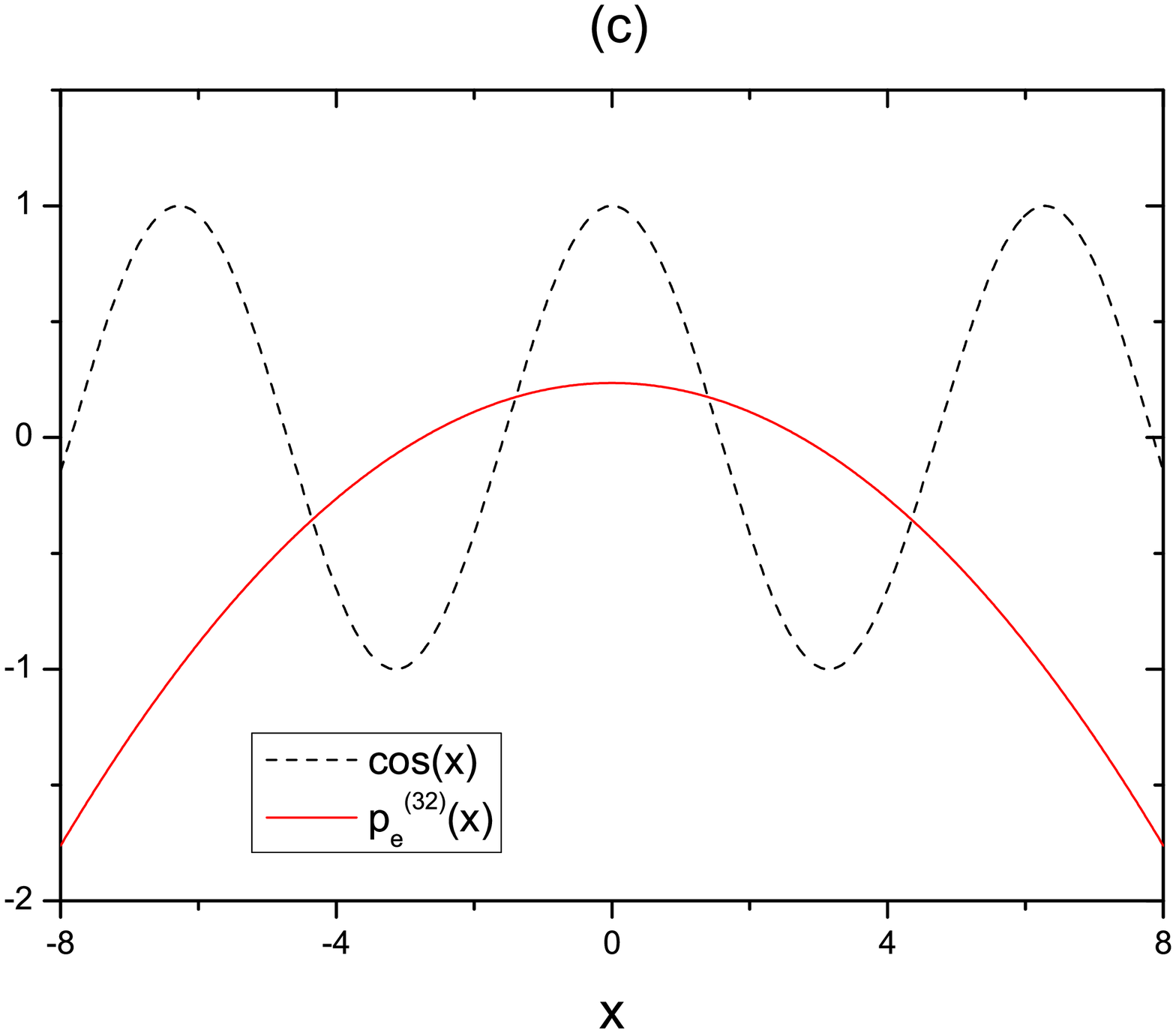}
 \centering\includegraphics[width=7.5cm]{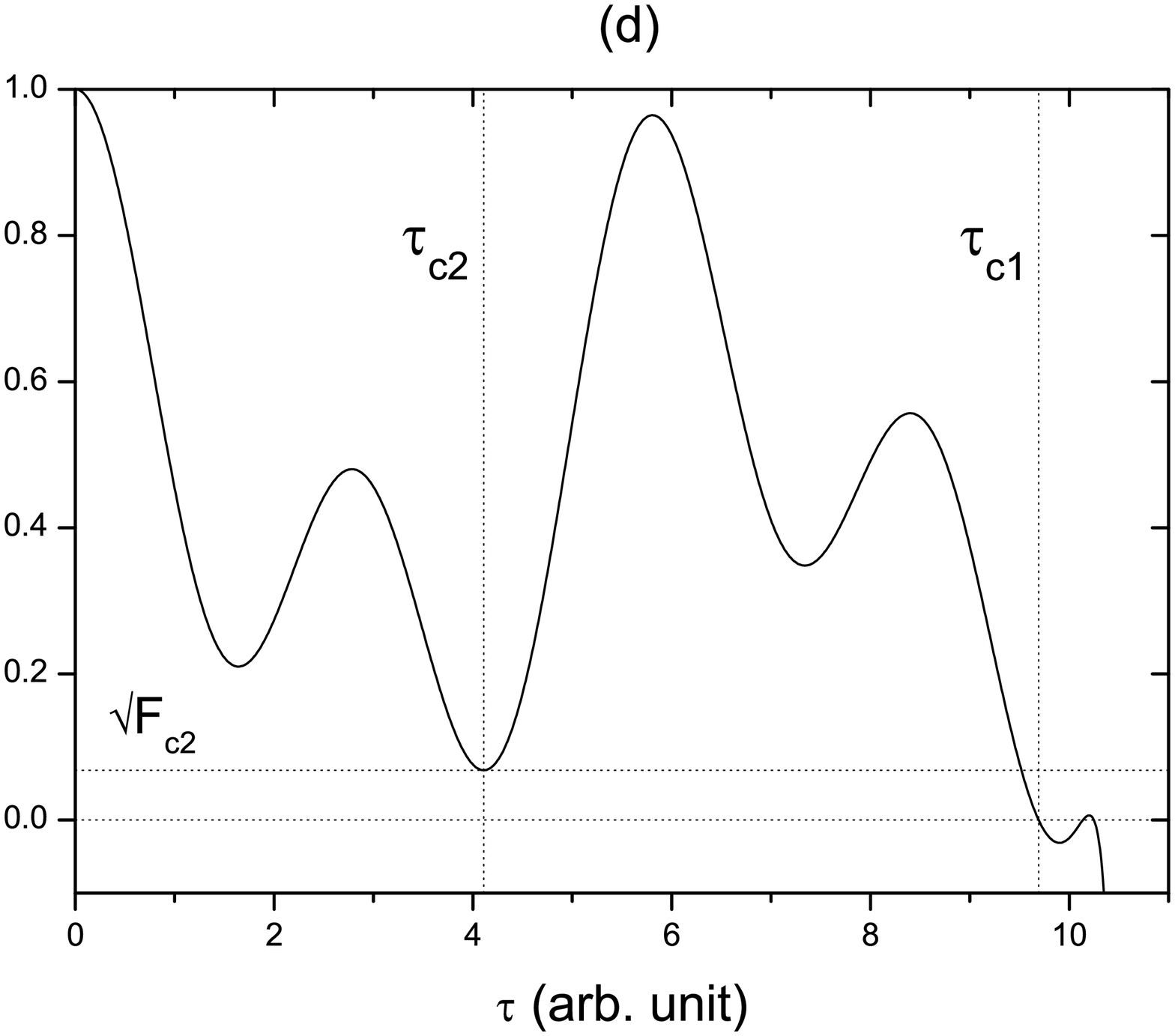}
 \caption{[color online.]  (a)~The time evolution of root fidelity $\sqrt{F}$
  for $|\varphi_e(0)\rangle$ under the Hamiltonian $H_e$.
  (b)~The polynomial $p_e(x)$ is a very good approximation of $\cos x$ for $0
  \le x \lesssim 22$.
  (c)~The curves $\cos(x)$ and $p_e^{(32)}(x)$ intersect at four distinct
  points only.
  (d)~A plot of the R.H.S. of inequality~\eqref{E:inequality_abs_E} with $p(x)
  = p_e(x)$ and $\absE{2k} = 4 \hbar^{2k} [1 + (11/5)^{2k} ] / 15$ for $k =
  1,2,\ldots,17$.
  \label{F:example}
 }
\end{figure*}

 Actually, this method can be adapted to show $p(x) \le \cos (x-\theta)$ for
 all $x\ge 0$ for a variety of polynomial $p(x)$ constructed out of Hermite
 interpolation.

\subsection{\QSL and a first-order phase transition}
\label{Subsec:FOPT}
 Since $p_e(x) \le \cos x$ for all $x \ge 0$, it induces the following \QSL:
\begin{widetext}
\begin{align}
 \sqrt{F} & \le 1 - 5.0000\times 10^{-1} \absE{2} \left( \frac{\tau}{\hbar}
  \right)^2 + 4.1667\times 10^{-2} \absE{4} \left( \frac{\tau}{\hbar} \right)^4
  - 1.3889\times 10^{-3} \absE{6} \left( \frac{\tau}{\hbar} \right)^6 \nonumber
  \\
 & \quad +2.4802\times 10^{-5} \absE{8} \left( \frac{\tau}{\hbar} \right)^8 
  -2.7557\times 10^{-7} \absE{10} \left( \frac{\tau}{\hbar} \right)^{10}
  + 2.0876\times 10^{-9} \absE{12} \left( \frac{\tau}{\hbar} \right)^{12}
  \nonumber \\
 & \quad -1.1469\times 10^{-11} \absE{14} \left( \frac{\tau}{\hbar}
  \right)^{14} + 4.7766\times 10^{-14} \absE{16} \left( \frac{\tau}{\hbar}
  \right)^{16} -1.5585\times 10^{-16} \absE{18} \left( \frac{\tau}{\hbar}
  \right)^{18} \nonumber \\
 & \quad + 4.0798\times 10^{-19} \absE{20} \left( \frac{\tau}{\hbar}
  \right)^{20}
  -8.6954\times 10^{-22} \absE{22} \left( \frac{\tau}{\hbar}
  \right)^{22} +1.5125\times 10^{-24} \absE{24} \left( \frac{\tau}{\hbar}
  \right)^{24} \nonumber \\
 & \quad -2.1160\times 10^{-27} \absE{26} \left( \frac{\tau}{\hbar}
  \right)^{26} + 2.2929\times 10^{-30} \absE{28} \left( \frac{\tau}{\hbar}
  \right)^{28} -1.7955\times 10^{-33} \absE{30} \left( \frac{\tau}{\hbar}
  \right)^{30} \nonumber \\
 & \quad + 8.9531\times 10^{-37} \absE{32} \left( \frac{\tau}{\hbar}
  \right)^{32} 
  -2.1140\times 10^{-40} \absE{34} \left( \frac{\tau}{\hbar} \right)^{34} .
 \label{E:complete_example_inequality}
\end{align}
\end{widetext}
 In particular, by putting $\absE{2k} = \aveE{2k} = 4 \hbar^{2k} [1 +
 (11/5)^{2k} ] / 15$ for $k = 1,2,\ldots,17$, this construction leads to a \QSL
 which gives a tight lower bound for the evolution time in the cases of
 $\sqrt{F} = \sqrt{F_{c1}}$ and $\sqrt{F_{c2}}$.  Fig.~\ref{F:example}d depicts
 that for root fidelity $\sqrt{F} = \sqrt{F_{c1}} = 0$, the corresponding \QSL
 is $\tau \in [\tau_{c1},10.138] \cup [10.248,+\infty)$.  Combined with the
 example of the evolution of $|\varphi_e(0)\rangle$, I conclude that whenever a
 quantum state with $\absE{2k}$ is equal to the value given in the previous
 paragraph for $k=1,2,\ldots ,17$, the minimum evolution time $\tau_{\min}$ for
 it to evolve to another state of root fidelity~$0$ is $\tau_{c1}$.  By
 gradually increasing $\sqrt{F}$, the allowable region for $\tau$ increases and
 $\tau_{\min} \le \tau_{c1}$.  Most importantly, by increasing $\sqrt{F}$ to
 $\sqrt{F_{c2}} \approx 0.0682$, the allowable $\tau$ becomes $\{ \tau_{c2} \}
 \cup [9.519,+\infty)$ with $\tau_{\min} = \tau_{c2}$.  That is, a new
 permissible time interval appears and a genuine forbidden evolution interval
 $(\tau_{c2},9.519)$ is formed.  Besides, $\tau_{\min}$ shows first-order phase
 transition at $\sqrt{F} = \sqrt{F_{c2}}$.  This is the first \QSL that
 captures this type of phase transition.
 Further significance of this result
 is reported in Appendix~\ref{App:significance}.

\subsection{The reverse problem}
\label{Subsec:reverse_problem}
 Note that the above method to construct a \QSL with genuine forbidden speed
 intervals is generic.  In fact, it brings us to the following reverse problem,
 which has never been studied before.  Given an initial state, an Hamiltonian
 and a required fidelity $F$, is it possible to find a \QSL whose minimum
 permissible evolution time equals the actual evolution time needed?
 By modifying the proof of the existence of $p(x)$, I show in
 Appendix~\ref{App:reverse_problem} that the answer is affirmative in
 finite-dimensional Hilbert space.

\section{Discussions and outlook}
\label{Sec:Outlook}
 To summarize, I have reported an efficient method to construct new \QSLs.  An
 important feature of this method is that by specifying a finite number of
 compatible observables in the form of various moments of energy of the system,
 the resultant \QSL is independent of the Hilbert space dimension.  Thus, the
 two most important consequences of this construction, namely, the existence of
 forbidden speed intervals and certain first-order phase transition are very
 strong results since they cannot come from an overly restricted set of
 constraints on a low-dimensional quantum system that almost fixing the
 Hamiltonian and the initial state.
 More importantly, this study opens up a more general research direction on the
 tradeoff between the amount of partial information describing a quantum system
 and the constraints on its information processing capability in which a lot of
 works can be done.

\medskip
\begin{acknowledgments}
 I thank F.\ K.\ Chow, C.-H.\ F.\ Fung and C.\ Y.\ Wong for their useful
 discussions.  This work is supported in part by the RGC Grants HKU~700712P and
 HKU8/CRF/11G of the Hong Kong SAR Government.
\end{acknowledgments}

\appendix

\section{Further Significance Of The \QSL Associated With The Degree~34
 Polynomial $\boldsymbol{p_e(x)}$ Reported In The Main Text
 \label{App:significance}}
 
 Recall that for fidelity $F = 0$, the \QSL reported in inequality~(4) in the
 main text leads to a minimum evolution time of $\tau_{c1}$ for $\absE{2k} =
 4\hbar^{2k} [1+(11/5)^{2k}]/15$ for $k = 1,2,\ldots,17$.  Further, this bound
 is tight for it can be achieved by the state $|\varphi_e(0)\rangle$ under the
 evolution of the Hamiltonian $H_e$.  Note that the highest and lowest energy
 eigenvector components of $|\varphi_e(0)\rangle$ are $|\pm 11\hbar/5\rangle$.
 Hence, the phase angle difference $\chi$ rotated during the time $\tau_{c1}$
 between these two components equals $22\tau_{c1}/5 > 2\pi$.  That is, the
 relative phase angle between two components has to rotate more than one
 complete circle in order to evolve $|\varphi_e(0)\rangle$ to its orthogonal
 complement.  This
 is a new situation for these relative phase angles rotated in all known
 \QSLs to date~\cite{PHYSCOMP96,ML98,GLM03,ZZ06,Chau10,LC13} are at most
 $2\pi$.  The implication is that for states obeying the above constraints on
 its various moments of energy, they cannot evolve to their orthogonal
 complement without some time of ``time wastage'' as some of the relative phase
 angle change must be greater than a complete circle.

 I also remark that this is the first tight \QSL with the property that the
 ``magic state'' saturating this \QSL in the case of $F = 0$ has to be at least
 four-dimensional.  The corresponding ``magic states'' for all previous \QSLs
 are at most three-dimensional~\cite{PHYSCOMP96,ML98,GLM03,ZZ06,Chau10,LC13}.
 The reason why the ``magic state'' saturating this \QSL is at least
 four-dimensional is as follows.  From the discussion on the conditions for
 equality of inequality~(3b) in the main text and the
 construction of $p_e$ that leads to the \QSL, the magic state
 $|\psi(0)\rangle$, expressed in the energy eigenbasis, must be in the form
 $\alpha_0 |0\rangle + \alpha_1 |E\rangle + \alpha_2 |-E\rangle + \alpha_3
 |11E/5\rangle + \alpha_4 |-11E/5\rangle$ with $E > 0$.  Furthermore, the
 evolution time to an orthogonal state equals $\tau_{c1}\hbar / E$.  For the
 given constraints in $\absE{2k}$'s, I arrive at $|\alpha_0|^2 = 7/15$,
 $|\alpha_1|^2 + |\alpha_2|^2 = |\alpha_3|^2 + |\alpha_4|^2 = 4/15$.  Consider
 the imaginary part of $\langle\psi(0)|\psi(\tau_{c1}\hbar / E)\rangle$, I have
 $(|\alpha_2|^2 - |\alpha_1|^2) \sin \tau_{c1} + (|\alpha_4|^2 - |\alpha_3|^2)
 \sin (11\tau_{c1}/5) = 0$.  Hence, at most one of the $\alpha_j$'s can be
 zero.  Thus, the ``magic state'' is at least four-dimensional.

\section{Existence Of A \QSL For The Reverse Problem For Finite-Dimensional
 Hilbert Space Systems \label{App:reverse_problem}}

 Denote the state at time $\tau$ under the evolution of the time-independent
 Hamiltonian $H$ in a $d$-dimensional Hilbert space by $|\varphi(\tau)\rangle$
 with $|\varphi(0)\rangle = \sum_{j=1}^d \alpha_j |E_j\rangle$.  Surely, the
 root fidelity at time $\tau$ is given by the continuous function
\begin{displaymath}
 \sqrt{F}(\tau) = \left| \langle\varphi(0)|\varphi(\tau)\rangle \right| =
 \sum_{j=1}^d |\alpha_j|^2 \cos \left[ \frac{E_j \tau}{\hbar} - \theta(\tau)
 \right]
\end{displaymath}
 where $\theta(\tau)$ is the argument of the complex number
 $\langle\varphi(0)|\varphi(\tau)\rangle$.  Although $\theta$ can only be
 determined modulo~$2\pi$, I may uniquely fix it by the integral curve
 describing the time evolution of the argument of
 $\langle\varphi(0)|\varphi(\tau)\rangle$ with initial condition $\theta(0) =
 0$ provided that $\tau$ is less than or equal to the first time when $\sqrt{F}
 = 0$.  And from now on, I assume $\theta(\tau)$ to be this smooth integral
 curve.

 I first write down several properties of the function $\sqrt{F}(\tau)$.
 Denote the first time when $\sqrt{F}(\tau)$ reaches a certain fixed value
 $\sqrt{F_0} \in [0,1]$ by $\tau_{\min}$.  Suppose further that $\tau_{\min}$
 is finite.  I define $\tau_\text{turn}$ and $\epsilon_0$ as follows.  Suppose
 $\sqrt{F}(\tau)$ is a decreasing function in $[0,\tau_{\min}]$, then I set
 $\tau_\text{turn} = 0$ and $\epsilon_0 = 1$.  Otherwise, since $\sqrt{F}$ is
 continuous and differentiable provided that $\sqrt{F} > 0$, there is a turning
 point in $[0,\tau_{\min})$.  Denote the turning point in $[0,\tau_{\min})$
 closest to $\tau_{\min}$ by $\tau_\text{turn}$.  Then, I set $\epsilon_0 =
 \min_{\tau\in [0,\tau_\text{turn}]} \sqrt{F}(\tau) - \sqrt{F_0}$.  It is
 well-defined because the minimum exists owning to the continuity of
 $\sqrt{F}$; and it is positive for otherwise $\tau_{\min}$ will not be the
 first time when the root fidelity reaches $\sqrt{F_0}$.  Note that no matter
 whether there is a turning point for $\sqrt{F}$ in $[0,\tau_{\min}]$ or not,
 $\sqrt{F}(\tau)$ is decreasing in $[\tau_\text{turn}, \tau_{\min}]$.  More
 importantly, for $0\le \tau \le \tau_{\min}$, $\sqrt{F}(\tau) < \sqrt{F_0} +
 \epsilon_0$ only in $[\tau_\text{turn},\tau_{\min}]$.  That is to say, the
 function $\sqrt{F}(\tau)$ is one-one in the domain $[\tau_{\text{turn}'},
 \tau_{\min}]$ and range $[\sqrt{F_0},\sqrt{F_0}+\epsilon_0]$, where
 $\tau_{\text{turn}'}$ is the closest point to $\tau_{\min}$ in
 $[0,\tau_{\min})$ with $\sqrt{F}(\tau_{\text{turn}'}) = \sqrt{F_0} +
 \epsilon_0$.  Last but not least, for a sufficiently small $\delta > 0$,
\begin{displaymath}
 \sqrt{F}(\tau) - \sqrt{F_0} = \mbox{O}((\tau_{\min} - \tau)^\beta)
\end{displaymath}
 for $\tau \in (\tau_{\min} - \delta,\tau_{\min})$ for some $\beta > 0$.

 Next, I consider the \QSL induced by a polynomial $p(x)$.  Using the idea in
 the main text, suppose $\cos x \ge p(x) = \sum_{k=0}^n c_k x^k$ for all $x \ge
 x_{\min} \equiv \min_j \min_{\tau\in [0,\tau_{\min}]} [E_j \tau / \hbar -
 \theta(\tau)]$.  Then, I have a \QSL in the form of an inequality
\begin{align}
 \sqrt{F}(\tau) & \ge \sum_{j=1}^d |\alpha_j|^2 p(\frac{E_j \tau}{\hbar} -
  \theta(\tau)) \nonumber \\
 & = \sum_{j=1}^d \sum_{k=0}^n c_k |\alpha_j|^2 \left[ \frac{E_j
  \tau}{\hbar} - \theta(\tau) \right]^k \nonumber \\
 & = \sum_{k=0}^n \sum_{\ell=0}^k c_k \binom{k}{\ell} \aveE{\ell}
  \theta(\tau)^{k-\ell} \left( \frac{\tau}{\hbar} \right)^\ell
 \label{E:induced_QSL}
\end{align}
 whenever $0\le\tau\le\tau_{\min}$.

 Now, I consider the set of polynomials ${\mathcal S}_{n,\gamma}$ with the
 properties that $p(x) \in {\mathcal S}_{n,\gamma}$ if and only if
\begin{itemize}
 \item $\deg p(x) \le n$;
 \item $p(x) \le \cos x$ whenever $x \ge x_{\min}$;
 \item $p(x_j) = \cos x_j$ for all $j$, where $x_j = E_j \tau_{\min} / \hbar -
  \theta(\tau_{\min})$; and
 \item $\cos x - p(x) = \mbox{O}((x-x_j)^\gamma)$ in the neighborhood of $x_j$
  for all $j$.
\end{itemize}
 Using the Hermite interpolating polynomial construction in the main text,
 I know that for each fixed $\gamma > 1$, the set
 ${\mathcal S}_{n,\gamma}$ is non-empty provided that $n$ is sufficiently
 large.  Clearly, ${\mathcal S}_{n,\gamma}$ is a convex set.  In addition, it
 is easy to see that the functional
\begin{displaymath}
 G[f_1,f_2] = \max_{x\in [x_{\min},x_{\max}]} \left| f_1(x) - f_2(x) \right|
\end{displaymath}
 is convex where $x_{\max} = \max_j [E_j \tau_{\min} / \hbar -
 \theta(\tau_{\min})]$.

 Note that for a fixed $\gamma > 1$, there is a sequence of polynomials $p_n(x)
 \in {\mathcal S}_{n,\gamma}$ such that $\lim_{n\to +\infty} G[p_n(x),\cos x] =
 0$.  In fact, each $p_n(x)$ can be chosen to be the optimal degree $\le n$
 polynomial in ${\mathcal S}_{n,\gamma}$ that minimizes the functional
 $G[\cos x,\cdot]$ via convex optimization~\cite{BV04}.

 With the above background preparation, I am ready to prove the existence of a
 polynomial $p(x)$ that solves the reverse problem.  That is, the \QSL induced
 by $p(x)$ in inequality~\eqref{E:induced_QSL} has the property that the
 smallest non-negative time at which the R.H.S. of this inequality is
 $\sqrt{F_0}$ occurs when $\tau = \tau_{\min}$ provided that $\aveE{k}$'s are
 set to the $k$th moment of the energy of the state $|\varphi(0)\rangle$.

 I choose a sufficiently small $0 < \delta < \tau_{\min} - \tau_\text{turn}$
 such that
\begin{displaymath}
 \sqrt{F}(\tau) - \sqrt{F_0} \ge \zeta (\tau_{\min} - \tau)^{\beta}
\end{displaymath}
 for all $\tau \in (\tau_{\min} - \delta,\tau_{\min}]$, where $\zeta > 0$.  For
 this $\delta > 0$, I can find sufficiently large $\gamma$ and $n$ such that
 ${\mathcal S}_{n,\gamma}$ is non-empty and the $p(x) \in
 {\mathcal S}_{n,\gamma}$ that minimizes the functional $G[\cos x,\cdot]$ obeys
 $G[\cos x,p(x)] < \epsilon$ where
\begin{displaymath}
 \epsilon \equiv \sqrt{F}(\tau_{\min} - \delta) - \sqrt{F_0} \in (0,
 \epsilon_0] .
\end{displaymath}

 I claim that the \QSL induced by this $p(x)$ solves the reverse problem.
 This is because by my construction, for $\tau = \tau_{\min}$, the R.H.S. of
 inequality~\eqref{E:induced_QSL} equals $\sum_j |\alpha_j|^2 p(x_j) = \sum_j
 |\alpha_j|^2 \cos x_j = \sqrt{F_0} = \sqrt{F}(\tau_{\min})$.

 I proceed to consider the case of $\tau\in [0,\tau_{\min} - \delta]$.  As
 $\cos x - p(x) \le G[\cos x,p(x)] < \epsilon$ for all $x > x_{\min}$, I
 conclude that
\begin{align*}
 \sqrt{F}(\tau) & = \sum_j |\alpha_j|^2 \cos \left[ \frac{E_j \tau}{\hbar} -
 \theta(\tau) \right]  \nonumber \\
 & < \epsilon + \sum_j |\alpha_j|^2 p(\frac{E_j\tau}{\hbar} - \theta(\tau)) .
\end{align*}
 Hence,
\begin{align*}
 \sum_j |\alpha_j|^2 p(\frac{E_j\tau}{\hbar} - \theta(\tau)) & > \sqrt{F}(\tau)
  - \epsilon \nonumber \\
 & \ge \sqrt{F}(\tau_{\min} - \delta) - \epsilon \nonumber \\
 & = \sqrt{F_0} .
\end{align*}
 In other words, the R.H.S. of inequality~\eqref{E:induced_QSL} greater than
 $\sqrt{F_0}$ in the time interval $[0,\tau_{\min} - \delta]$.

 Finally, I consider the case of $\tau \in (\tau_{\min} - \delta,\tau_{\min})$.
 I have
\begin{align*}
 & \sqrt{F_0} + \zeta (\tau_{\min} - \tau)^\beta \nonumber \\
 \le & \sqrt{F}(\tau) \nonumber \\
 = & \sum_j |\alpha_j|^2 p(\frac{E_j\tau}{\hbar} - \theta(\tau)) +
 \mbox{O}((\Delta_j(\tau))^\gamma)
\end{align*}
 where $\Delta_j(\tau) = E_j (\tau_{\min} - \tau) / \hbar + \theta(\tau) -
 \theta(\tau_{\min})$.  Since $\theta(\tau)$ is smooth in this time interval,
 $\Delta_j(\tau) = \mbox{O}((\tau_{\min} - \tau)^{\gamma'})$ for some $\gamma'
 > 0$.  Therefore,
\begin{align*}
 & \sum_j |\alpha_j|^2 p(\frac{E_j\tau}{\hbar} - \theta(\tau)) \nonumber \\
 \ge & \sqrt{F_0} +
 \zeta (\tau_{\min} - \tau)^\beta + \mbox{O}((\tau_{\min} - \tau)^{\gamma
 \gamma'}) .
\end{align*}
 Since $\zeta > 0$, by picking a sufficiently small $\delta > 0$ and a
 sufficiently large $\gamma$ so that $\gamma\gamma' > \beta$ (and a sufficient
 large $n$ so that ${\mathcal S}_{n,\gamma}$ is non-empty), the R.H.S. of
 inequality~\eqref{E:induced_QSL} is greater than $\sqrt{F_0}$ in this time
 interval.

 To summarize, the \QSL induced by any $p(x)$ in this ${\mathcal S}_{n,\beta}$
 is a solution of the reverse problem.  This completes the proof of my claim.

 Lastly, let me make the following remark.  Suppose $0\le \tau_1 < \tau_2 <
 \ldots < \tau_\ell$ are $\ell$ distinct numbers with $\sqrt{F}(\tau_j) =
 \sqrt{F_0}$ for all $j$, where $\sqrt{F}(\tau)$ is the root fidelity between
 $|\varphi(\tau)\rangle$ and $|\varphi(0)\rangle$ under the action of a
 time-independent Hamiltonian $H$.  Then, it is not difficult to adapt the
 above procedure to construct a polynomial whose induced \QSL has the
 properties that
\begin{itemize}
 \item the induced \QSL is an equality at times $\tau_1, \ldots , \tau_\ell$
  provided that $\sqrt{F} = \sqrt{F_0}$ and $\aveE{k}$ is the $k$th moment of
  the average energy of the state $|\varphi(0)\rangle$;
 \item the induced \QSL is a strict inequality at time $\tau \in [0,\tau_\ell]$
  provided that $\sqrt{F}(\tau) > F_0$.
\end{itemize}
 The proof is left to the interested readers.

\bibliographystyle{apsrev4-1}
\bibliography{qc59.5}

\end{document}